\documentclass[a4paper,12pt]{article}

\usepackage{amsmath,amssymb,amsthm}
\usepackage[dvipdfmx]{graphicx}
\usepackage[margin=24mm]{geometry}
\usepackage[doublespacing]{setspace}
\usepackage{enumerate}
\usepackage[round]{natbib}

\newcommand{\keywords}[1]{\par\noindent\textbf{Keywords:}%
 #1}
\newcommand{\jel}[1]{\par\noindent\textbf{JEL classification codes:}%
 #1}
\newcommand{\rnd}[1]{\left(#1\right)}
\newcommand{\crly}[1]{\left\{#1\right\}}

\newcommand{\set}[2]{\left\{#1 \mid #2\right\}}
\newcommand{\map}[3]{#1: #2 \to #3}
\newcommand{\card}[1]{\left|#1\right|}

\newtheorem{proposition}{Proposition}

\newtheorem{lemma}{Lemma}

\theoremstyle{definition}
\newtheorem{definition}{Definition}

\theoremstyle{remark}
\newtheorem{remark}{Remark}

\title
{
 Extractive Contest Design\thanks
 {
  The author is grateful to
  Kazuo Yamaguchi
  for his valuable comments.
  This work was supported by JSPS KAKENHI Grant Number JP19K01563.
 }
}
\author
{
 Tomohiko Kawamori\thanks
 {
  Faculty of Economics,
  Meijo University,
  1-501 Shiogamaguchi, Tempaku-ku, Nagoya 468-8502, Japan.
  {\tt kawamori@meijo-u.ac.jp}
 }
}
\date
{
}

\begin{document}

\maketitle

\abstract
{
 We consider contest success functions (CSFs) that extract contestants' prize values.
 In the common-value case,
 there exists a CSF extractive in any equilibrium.
 In the observable-private-value case,
 there exists a CSF extractive in some equilibrium;
 there exists a CSF extractive in any equilibrium
 if and only if the number of contestants is greater than or equal to three or the values are homogeneous.
 In the unobservable-private-value case,
 there exists no CSF extractive in some equilibrium.
 When extractive CSFs exist,
 we explicitly present one of them.
}

\keywords
{
 contest success function;
 extraction of values;
 common or private values;
 observable or unobservable values;
 aggregate effort equivalence across equilibria
}

\jel
{
 C72;
 D72
}

\newpage
\section{Introduction}\label{sec:introduction}

In the literature on contests,
formalized by \cite{Tullock1980},
the design of contest success functions (CSFs) that maximize the aggregate effort of contestants is a key topic.
Many papers have examined this issue.
Maximization of aggregate effort provides CSFs with a positive foundation.\footnote
{
 \cite{Jia2013} referred to this as the \emph{optimally derived foundation}, one among the four types of foundations.
}${}^,$\footnote
{
 Some papers have provided CSFs with axiomatic foundations
 (e.g., \cite{Skaperdas1996} and \cite{Clark1998}).
}
In the rent-seeking interpretation,
the contest designer (politician) intends to maximize the contestants' efforts.
This is because
if the efforts are political contributions,
he/she obtains monetary benefits from these efforts,
and
if the efforts are political lobbying,
he/she flaunts his/her power through the efforts.
Thus,
he/she should determine the CSF as it maximizes the efforts.

This paper considers the design of CSFs that extract the contestants' values.
Most of the existing literature has investigated the maximization of aggregate effort in a class of CSFs that satisfy some restrictions.
Owing to these restrictions,
the aggregate effort is generically less than the highest value of the contestants,
i.e.,
the optimal CSF does not extract contestants' values.
Instead,
this paper
considers the design of CSFs without restrictions
and examines the extraction of contestants' values.

We define extractiveness of CSFs.
In a contest,
contestants make an effort,
and
the winner of a prize is determined according to a probability distribution that depends on the efforts.
A function that maps effort tuples to winning probability distributions is called a CSF.
We consider the design of CSFs that extract the contestants' prize values through the contestants' efforts.
We say that a CSF is \emph{extractive}
if
under this CSF,
there exists a Nash equilibrium such that
the aggregate effort is equal to the maximum value (the maximum of the contestants' prize values).
We also say that a CSF is \emph{strictly extractive}
if
this CSF is extractive,
and
under this CSF,
in every Nash equilibrium,
the aggregate effort is equal to the maximum value.
We consider pure-strategy Nash equilibria.

We examine extractiveness of CSFs,
focusing on
whether the contestants' prize values are common or private
and whether they are observable or unobservable by the contest designer.
First,
we consider the case where these values are common.
We consider the CSF such that
the winning probabilities are proportional to the $\frac{a}{a - 1}$th power of the efforts,
where $a$ is in $\crly{2,3,\dots,n}$ ($n$ is the number of contestants).
We show that
the CSF with $a = 2$ is strictly extractive for all common values.
Because this CSF does not depend on common values,
this result holds regardless of whether they are observable or unobservable.
For every common value,
we also show that
the CSF with $a \geq 3$ is extractive,
but not strictly extractive.
Thus,
aggregate effort equivalence across Nash equilibria
holds for $a = 2$,
but not for $a \geq 3$.
Second,
we consider the case where the values are private and observable.
For every value tuple,
we present an extractive CSF.
For every value tuple,
we show that
there exists a strictly extractive CSF
if and only if
the number of contestants is greater than or equal to $3$
or the values are homogeneous.
When strictly extractive CSFs exist,
we present one of them.
The aggregate effort equivalence across Nash equilibria holds
if and only if
the number of contestants is greater than or equal to $3$
or the values are homogeneous.
Third,
we consider the case where the values are private and unobservable.
We show that
there exists no CSF that is extractive for all value tuples.
Therefore,
observability matters in the private-value case,
but not in the common-value case.

Several papers have presented CSFs that are reduced to CSFs extractive in the unobservable-common-value case.
In the $2$-contestant case,
\cite{Nti2004} (\cite{Epstein2013}; \cite{Ewerhart2017}, resp.)\ presented a CSF that maximizes the aggregate effort in a class of CSFs
(Section 4 (Subsection 4.2; Proposition 6, resp.)).
In the $n$-contestant common-value case,
\cite{Michaels1988} did so
(Subsection 2.1).
Each CSF in \cite{Nti2004}, \cite{Epstein2013} and \cite{Ewerhart2017} in the common-value case (the CSF in \cite{Michaels1988}, resp.)\ is the CSF such that the winning probabilities are proportional to the $2$nd ($\frac{n}{n - 1}$th, resp.)\ power of the efforts,
and
it is in the unobservable-common-value case
because it does not depend on the common value.
The CSF using the $2$nd power is strictly extractive in the $2$-contestant case.
We show that
this CSF is strictly extractive in the $n$-contestant case.
The CSF using the $\frac{n}{n - 1}$th power is extractive.
We show that
if $n \geq 3$,
this CSF is not strictly extractive.
\cite{Perez-Castrillo1992} derived Nash equilibria under the CSF such that the winning probabilities are proportional to the $r$th power of the efforts in the $n$-contestant case
(Proposition 4).
The result of \cite{Perez-Castrillo1992} implies that
CSFs such that the winning probabilities are proportional to the $\frac{a}{a - 1}$th power of the efforts ($2 \leq a \leq n$) are extractive.
We show that
if $a \geq 3$,
this CSF is not strictly extractive.

Several papers have presented CSFs that are extractive but not strictly extractive in the observable-value case.
In the $2$-contestant common-value case,
\cite{Glazer1993} presented a CSF such that
a certain contestant wins if his/her effort is equal to his/her value,
and the other contestant wins otherwise
(Subsection 3.1).
In the $2$-contestant case (the $n$-contestant case, resp.),
\cite{Nti2004} (\cite{Franke2018}, resp.)\ presented a CSF such that
a contestant with the maximum value wins if his/her effort is greater than or equal to his/her value,
and a contestant with the second highest value, which may be equal to the maximum value, wins otherwise
(Proposition 2 (Proposition 4.7, resp.)).
These CSFs are extractive.
However,
they are not strictly extractive,
because there exists a Nash equilibrium such that
every contestant's effort is zero.
\cite{Nti2004} suggested that
under a modified CSF such that the effort threshold is slightly lowered,
the Nash equilibrium such that every contestant's effort is zero is removed.
However,
under the modified CSF,
in a unique Nash equilibrium,
the aggregate effort is slightly smaller than the maximum value.
Meanwhile,
in the $3$-or-more-contestant or homogeneous-value case,
we present strictly extractive CSFs.
In any Nash equilibrium under such CSFs,
the aggregate effort is exactly equal to the maximum value.
In the $2$-heterogeneous-contestant case,
we show that there exists no strictly extractive CSF.
This implies that
the contest designer cannot design a strictly extractive CSF
even though he/she can fully use the information of the values.

Several papers have shown the extraction of values in mixed-strategy Nash equilibria.
\cite{Hillman1989} (\cite{Che1998}; \cite{Baye1993}; \cite{Baye1996}, resp.)\ showed that
in the all-pay auction,
if the highest two values are equal,
the expected aggregate effort is equal to the maximum value
in any mixed-strategy Nash equilibrium
(Proposition 1; the second last paragraph in Section 3 (equation (9); Theorem 1, resp.)).
\cite{Ewerhart2017a} showed it in a modified all-pay auction
(Proposition 1).
\cite{Alcalde2010} showed that
under CSFs that satisfy certain conditions,
there exists a mixed-strategy Nash equilibrium such that
if the highest two values are equal,
the expected aggregate effort is equal to the maximum value 
(Theorem 3.2).
We show the extraction of values in pure-strategy Nash equilibria,
even if the highest two values are not equal.

Several papers have considered maximizing aggregate effort in a class of CSFs.
CSFs with the following devices have been examined:
concave technologies\footnote
{
 A concave technology is a concave function that transforms efforts.
 Winning probabilities are determined based on the transformed efforts.
}
in the lottery contest
(\cite{Dasgupta1998});
concave technologies and power technologies\footnote
{
 A power technology is a power function that transforms efforts.
}
in the lottery contest
(\cite{Nti2004});
power technologies in the lottery contest
(\cite{Michaels1988});
biases multiplying efforts in the lottery contest
(\cite{Franke2013});
biases multiplying efforts with power technologies in the lottery contest and biases multiplying efforts in the all-pay auction
(\cite{Epstein2013});\footnote
{
 \cite{Epstein2011} considered the same class of CSF,
 but a different objective of the contest designer,
 which is the weighted sum of the aggregate effort and welfare.
}
biases multiplying efforts in the lottery contest and all-pay auction
(\cite{Franke2014});
head starts added to efforts in the lottery contest and all-pay auction
(\cite{Franke2014a});
biases multiplying efforts given a power technology in the lottery contest
(\cite{Ewerhart2017});
biases multiplying efforts and head starts added to efforts
in the lottery contest and all-pay auction
(\cite{Franke2018}).
\cite{Fang2002}
compared the simple lottery contest with the simple all-pay auction.
Owing to restrictions on the forms of CSFs,
the maximized aggregate effort is not equal to the maximum value except for the aforementioned results.
In our paper,
because no restrictions are imposed on the forms of CSFs,
the values are extracted.

Several papers have considered the aggregate effort under asymmetric information.
In
\cite{Kirkegaard2012},
\cite{Perez-Castrillo2016},
\cite{Matros2016},
\cite{Drugov2017}
and \cite{Olszewski2020},
the values or productivities of the efforts are private information.
In our paper,
contestants know the contestants' values;
the contest designer knows them in the observable-value case,
but not in the unobservable-value case.

The contributions of this paper are as follows.
In the observable-private-value case,
we present strictly extractive CSFs in the $3$-or-more-contestant or homogeneous-value subcase,
where
we only use the pure-strategy Nash equilibria,
whereas we show that
there exists no strictly extractive CSF in the other subcase.
In the unobservable-private-value case,
we show that
there exists no extractive CSF.
In the common-value case,
we show that
the CSF with the $2$nd-power technology is strictly extractive in the multi-contestant contest,
and
the CSF with the $\frac{a}{a - 1}$th-power technology ($a \geq 3$) is not strictly extractive.
We demonstrate that
for extractive or strictly extractive CSFs to exist,
observability of values matters in the private-value case,
but not in the common-value case.
The framework in this paper could serve as a general framework for investigating the extraction of values in contests.

The remainder of this paper is organized as follows.
Section \ref{sec:model} describes the model.
Section \ref{sec:results} presents the results.
Section \ref{sec:conclusion} concludes the paper.
The proofs of all the propositions are provided in the appendix.

\section{Model}\label{sec:model}

For
any sets $X$, $Y$ and $I$,
any $\map{f}{X}{Y^I}$
and any $x \in X$ and $i \in I$,
let $f_i\rnd{x}$ be the value of $f\rnd{x}$ for $i$.

Let $N$ be a finite set such that $\card{N} \geq 2$:
$N$ is the set of contestants.
Let $n := \card{N}$.
Let $X := \mathbb R_{\geq 0}^N$:
$X$ is the set of tuples of contestants' efforts.
Let $\Delta := \set{p \in \mathbb R_{\geq 0}^N}{\sum_{i \in N} p_i = 1}$:
the set of tuples of contestants' success probabilities
(for any $p \in \Delta$ and any $i \in N$,
$p_i$ is the probability of contestant $i$'s winning).
Let
$F := \set{f}{\map{f}{X}{\Delta}}$:
$F$ is the set of contest success functions (CSFs).
Let $V := \mathbb R_{> 0}^N$:
the set of tuples of contestants' prize values.
For any $f \in F$ and any $v \in V$,
let $\map{u^{f v}}{X}{\mathbb R^N}$ such that
for any $x \in X$ and any $i \in N$,
$u_i^{f v}\rnd{x} = f_i\rnd{x} v_i - x_i$:
$u_i^{f v}\rnd{x}$ is contestant $i$'s utility from effort tuple $x$
($f_i\rnd{x} v_i$ is the expected value that he/she obtains,
and $x_i$ is the cost of his/her effort).

For any $f \in F$ and any $v \in V$,
$\rnd{N,X,u^{f v}}$ is a strategic-form game:
$N$ is the set of players,
$X$ is the set of strategy tuples,
and $u^{f v}$ is the function that maps each strategy tuple to its payoff tuple.
For any $f \in F$ and any $v \in V$,
let $E^{f v}$ be the set of pure-strategy Nash equilibria in $\rnd{N,X,u^{f v}}$.
In the following,
we refer to a pure-strategy Nash equilibrium simply as a Nash equilibrium.

Let $\hat V := \set{v \in V}{\rnd{\forall i,j \in N} v_i = v_j}$:
the set of value tuples such that all contestants have a common value.
For any $v \in V$,
let
$m^v:= \max_{i \in N} v_i$
and $M^v := \arg\max_{i \in N} v_i$:
$m^v$ is the maximum of contestants' prize values,
and $M^v$ is the set of contestants with the maximum value.

\section{Results}\label{sec:results}

We refer to
the case where the domain of the value tuples is $\hat V$ ($V$, resp.)\ 
as
the \emph{common-value case} (\emph{private-value case}, resp.).
We also refer to
the case where value tuples are observable (unobservable, resp.)\ by the contest designer, i.e., CSFs can (cannot, resp.)\ depend on value tuples
the \emph{observable-value case} (\emph{unobservable-value case}, resp.).
We seek CSFs under which the equilibrium aggregate effort is equal to the maximum value
in
the common-value case,
the observable-private-value case
and the unobservable-private value case,
respectively.
We say that
in the private-value case,
if $v \in V$ satisfies that
for any $i,j \in N$,
$v_i = v_j$,
we say that
\emph{$v$ is homogeneous}.
The value tuples are observable by the contestants.

\subsection{Bound of aggregate effort}

For any CSF and any value tuple,
in any Nash equilibrium,
the aggregate effort is less than or equal to the maximum value.
\begin{proposition}\label{prop:bound_of_total_effort}
 Let
 $f \in F$
 and $v \in V$.
 Let $x^\ast \in E^{f v}$.
 Then,
 $\sum_{i \in N} x_i^\ast \leq m^v$.
\end{proposition}

We say that
a CSF is extractive
if
in some Nash equilibrium,
the aggregate effort is equal to the maximum value.
We say that
a CSF is strictly extractive
if
it is extractive
and
in any Nash equilibrium,
the aggregate effort is equal to the maximum value.
\begin{definition}\label{def:extractiveness}
 Let
 $f \in F$
 and $v \in V$.
 $f$ is \emph{extractive} for $v$
 if
 there exists $x^\ast \in E^{f v}$ such that
 $\sum_{i \in N} x_i^\ast = m^v$.
 $f$ is \emph{strictly extractive} for $v$
 if
 $f$ is extractive for $v$
 and
 for all $x^\ast \in E^{f v}$,
 $\sum_{i \in N} x_i^\ast = m^v$.
\end{definition}

\subsection{Common-value case}\label{subsec:common_values}

We consider the case where
contestants have a common value.
For any $a \in \mathbb N$ such that $2 \leq a \leq n$,
let $f^a \in F$ such that
for any $i \in N$ and any $x \in X$,
\begin{align*}
 f_i\rnd{x}
 = \begin{cases}
    \frac{x_i^{\frac{a}{a - 1}}}{\sum_{j \in N} x_j^{\frac{a}{a - 1}}} & \text{if $\rnd{\exists j \in N} x_j > 0$}\\
    \frac{1}{n} & \text{otherwise}.
   \end{cases}
\end{align*}
Under $f^a$,
winning probabilities are proportional to the $\frac{a}{a - 1}$th power of the efforts.

There exists a CSF strictly extractive for all common values.
\begin{proposition}\label{prop:strict_extractiveness_for_common_values}
 There exists $f \in F$ that is strictly extractive for  all $v \in \hat V$.
\end{proposition}
\begin{remark}
 A purely logical consequence of this proposition is that
 for all $v \in \hat V$,
 there exists $f \in F$ that is strictly extractive for $v$.
 Thus,
 this proposition implies that
 whether the contest designer can observe the common value,
 there exists a strictly extractive CSF.
 Furthermore,
 regardless of the observability,
 there exists an extractive CSF.
\end{remark}
\begin{remark}
 In the proof,
 such CSF $f$ is constructed as $f = f^2$.
\end{remark}
\begin{remark}
 As seen in the proof,
 under
 $f = f^2$
 and any $v \in \hat V$,
 for any $x \in X$,
 $x \in E^{f v}$
 if and only if
 for some $A \in 2^N$ such that $\card A = 2$,
 for any $i \in A$,
 $x_i = \frac{m^v}{2}$
 and
 for any $i \in N \setminus A$,
 $x_i = 0$.
\end{remark}
\begin{remark}
 Under
 $f = f^2$
 and any $v \in \hat V$,
 the aggregate effort equivalence across Nash equilibria holds.
\end{remark}

For all common values,
$f^a$ ($3 \leq a \leq n$) is extractive but not strictly extractive.
\begin{proposition}\label{prop:extractiveness_for_common_values_of_non-quadratic_power_technologies}
 Let $a \in \mathbb N$ such that $3 \leq a \leq n$.
 Let $v \in \hat V$.
 Then,
 $f^a$ is extractive for $v$ and not strictly extractive for $v$.
\end{proposition}
\begin{remark}\label{rem:extractiveness_for_unobservable_values_with_common_values}
 As seen in the proof,
 under $f = f^a$
 and any $v \in \hat V$,
 for any $x \in X$,
 $x \in E^{f v}$
 if
 (i)
 for some $A \in 2^N$ such that $\card A = a$,
 for any $i \in A$,
 $x_i = \frac{m^v}{a}$
 and
 for any $i \in N \setminus A$,
 $x_i = 0$,
 or
 (ii)
 for some $A \in 2^N$ such that $\card A = a - 1$,
 for any $i \in A$,
 $x_i = \frac{m^v a \rnd{a - 2}}{\rnd{a - 1}^3}$
 and
 for any $i \in N \setminus A$,
 $x_i = 0$.
 The aggregate effort in a strategy tuple satisfying (ii) is
 $
  \rnd{a - 1} \frac{v a \rnd{a - 2}}{\rnd{a - 1}^3}
  = \frac{v a \rnd{a - 2}}{\rnd{a - 1}^2}
  < v
 $.
 However,
 for any $a \in \mathbb N^{\mathbb N}$ such that
 $3 \leq a_n \leq n$
 and $\lim_{n \to \infty} a_n = \infty$,
 $
  \lim_{n \to \infty} \frac{v a_n \rnd{a_n - 2}}{\rnd{a_n - 1}^2}
  = v
 $.
\end{remark}
\begin{remark}
 Under
 $f = f^a$
 and any $v \in \hat V$,
 the aggregate effort equivalence across Nash equilibria
 does not hold.
\end{remark}

Under $f^a$ ($2 \leq a \leq n$),
each contestant's effort $x$ is transformed into $x^{\frac{a}{a - 1}}$,
and the winning probabilities are determined in proportion to the transformed efforts.
As $a$ is larger,
elasticity of the transformed effort $x^{\frac{a}{a - 1}}$ to effort $x$,
$\frac{d x^{\frac{a}{a - 1}}/d x}{x^{\frac{a}{a - 1}}/x}
= \frac{a}{a - 1}$,
is smaller;
thus,
in the Nash equilibrium such that the aggregate effort is equal to the maximum value,
the effort of each active contestant is smaller,
but
the number of active contestants is larger.

The results in \cite{Nti2004}, \cite{Epstein2013} and \cite{Ewerhart2017} imply that
$f^2$ is strictly extractive for all $v \in \hat V$ when $n = 2$.
Our paper shows that
$f^2$ is strictly extractive for all $v \in \hat V$ when $n \geq 3$.
The results in \cite{Michaels1988} imply that
$f^n$ is extractive for all $v \in \hat V$.
Our paper shows that
when $n \geq 3$,
for any $v \in \hat V$,
$f^n$ is not strictly extractive for $v$.
The results in \cite{Perez-Castrillo1992} imply that
$f^a$ is extractive for all $v \in \hat V$.
Our paper shows that
when $a \geq 3$,
for any $v \in \hat V$,
$f^a$ is not strictly extractive for $v$.

\subsection{Observable-private-value case}

We consider the case where
contestants have private values,
and the contest designer can observe them and design CSFs dependent on them.

For all value tuples,
there exists a CSF extractive for this value tuple.
\begin{proposition}\label{prop:extractiveness_for_observable_private_values}
 Let $v \in V$.
 Then,
 there exists $f \in F$ that is extractive for $v$.
\end{proposition}
\begin{remark}\label{rem:construction_of_extractive_csf_for_observable_private_values}
 In the proof,
 such CSF $f$ is constructed as follows.
 In the case where $n \geq 3$,
 for some distinct $i,j,k \in N$ such that $i \in M^v$,
 for any $x \in X$,
 if $x_i = m^v$,
 then $f_i\rnd{x} = 1$;
 if $x_i \neq m^v$ and $x_j > 0$,
 then $f_j\rnd{x} = 1$;
 if $x_i \neq m^v$ and $x_j = 0$,
 then $f_k\rnd{x} = 1$.
 In the case where $n = 2$ and $v \in \hat V$,
 $f$ is $f^2$ defined in Subsection \ref{subsec:common_values}.
 In the case where $n = 2$ and $v \notin \hat V$,
 for some $i \in M^v$,
 for any $x \in X$,
 $f_i\rnd{x} = \mathbf 1_{x_i = m^v}$.
\end{remark}

For all value tuples,
if
the number of contestants is greater than or equal to $3$
or
the values of the contestants are homogeneous,
there exists a CSF strictly extractive for this value tuple,
and
otherwise,
there does not.
\begin{proposition}\label{prop:strict_extractiveness_for_observable_private_values}
 Let $v \in V$.
 Then,
 there exists $f \in F$ that is strictly extractive for $v$
 if and only if
 $n \geq 3$ or $v \in \hat V$.
\end{proposition}
\begin{remark}
 In the proof,
 such $f$ is constructed as follows.
 In the case where $n \geq 3$ and the case where $n = 2$ and $v \in \hat V$,
 $f$ is one in the corresponding case in Remark \ref{rem:construction_of_extractive_csf_for_observable_private_values}.
\end{remark}
\begin{remark}
 Aggregate effort equivalence across Nash equilibria
 holds
 if and only if $n = 3$ or $v \in \hat V$.
\end{remark}

The above CSFs make the contestant with the maximum value win with certainty (or probability $\frac{1}{2}$)
if his/her effort is equal to the maximum value (or $\frac{1}{2}$ of the maximum value),
in order that
the aggregate effort is equal to the maximum value.
In the $3$-or-more-contestant or homogeneous-value case,
the above CSFs are designed
as they exclude the Nash equilibrium such that
every contestant's effort is zero.
In the other case,
it is impossible to design an extractive CSF that excludes such a Nash equilibrium
even though the contest designer can fully use the information on the contestants' values.

In the observable-value case,
\cite{Glazer1993}, \cite{Nti2004} and \cite{Franke2018} presented a CSF that is extractive,
but this is not strictly extractive.
For any $v \in V$,
when $n \geq 3$ or $v \in \hat V$,
our paper presents a CSF that is strictly extractive for $v$;
when $n = 2$ and $v \notin \hat V$,
our paper shows that
there exists no CSF that is strictly extractive for $v$.

\subsection{Unobservable-private-value case}

We consider the case where
contestants have private values,
and the contest designer cannot observe them and must design CSFs independent of them.

There exists no CSF that is extractive for all value tuples.
\begin{proposition}\label{prop:extractiveness_for_unobservable_private_values}
 There exists no $f \in F$ that is extractive for all $v \in V$.
\end{proposition}
\begin{remark}
 An immediate consequence of this proposition is that
 there exists no $f \in F$ that is strictly extractive for all $v \in V$.
\end{remark}

Under any CSF,
if
for some value tuple,
the aggregate effort in a Nash equilibrium is equal to the maximum value,
then
for some other value tuple,
it must be less than the maximum value.

In the literature,
it has not been examined
whether there exists a CSF that is extractive for all $v \in V$.
Our paper provides a negative answer to this question.

\section{Conclusion}\label{sec:conclusion}

Table \ref{tab:extractive_or_strictly_extractive_csfs} summarizes the results of this paper,
where
$\phi^{\mathrm{E}}$ is a formula meaning that
$f$ is extractive for $v$,
and
$\phi^{\mathrm{SE}}$ is a formula meaning that
$f$ is strictly extractive for $v$.
Whether the values are common or private is represented by whether the domain of the values is $\hat V$ or $V$.
Whether the values are observable or unobservable is represented by whether the order of the quantifiers is $\rnd{\forall v} \rnd{\exists f}$ or $\rnd{\exists f} \rnd{\forall v}$.
In the common-value case,
regardless of the observability,
there exist extractive and strictly extractive CSFs.
In the observable-private-value case,
there exists an extractive CSF,
but there does not always exist a strictly extractive CSF.
In the unobservable-private-value case,
there exists neither extractive nor strictly extractive CSF.
In the common-value case,
we also present a class of extractive CSFs that
can be used to control the number of active contestants.
\begin{table}[htbp]
 \begin{center}
  \begin{tabular}{|c|c|c|} \hline
   & Observable & Unobservable \\
   & $\rnd{\forall v} \rnd{\exists f}$ & $\rnd{\exists f} \rnd{\forall v}$ \\ \hline
   Common  & $\rnd{\forall v \in \hat V} \rnd{\exists f \in F} \phi^\mathrm{E}$ & $\rnd{\exists f \in F} \rnd{\forall v \in \hat V} \phi^\mathrm{E}$ \\
   $\hat V$ & $\rnd{\forall v \in \hat V} \rnd{\exists f \in F} \phi^\mathrm{SE}$ & $\rnd{\exists f \in F} \rnd{\forall v \in \hat V} \phi^\mathrm{SE}$ \\ \hline
   Private & $\rnd{\forall v \in V} \rnd{\exists f \in F} \phi^\mathrm{E}$ & $\neg \rnd{\exists f \in F} \rnd{\forall v \in V} \phi^\mathrm{E}$ \\
   $V$ & $\rnd{\forall v \in V} \rnd{\rnd{\exists f \in F} \phi^\mathrm{SE} \leftrightarrow n \geq 3 \vee v \in \hat V}$ & $\neg \rnd{\exists f \in F} \rnd{\forall v \in V} \phi^\mathrm{SE}$ \\ \hline
  \end{tabular}
 \end{center}
 \caption{Existence of extractive or strictly extractive CSFs}\label{tab:extractive_or_strictly_extractive_csfs}
\end{table}

In the unobservable-private-value case,
there exists no extractive CSF.
In such a case,
it is necessary to derive CSFs that maximize the expectation of the aggregate effort under some belief on value tuples.
For example,
this problem is formalized as follows:
\begin{align*}
 \max_{\rnd{f,x} \in F \times X^V}&\:
 \int_{v \in V} \sum_{i \in N} x_i\rnd{v} d P\rnd{v}\\
 \text{s.t.}&\:
 \rnd{\forall v \in V} x\rnd{v} \in E^{f v}
 \wedge \text{$x$ is measurable},
\end{align*}
where $P$ is a cumulative distribution function on $V$ (the designer's belief on value tuples).

In this paper,
we consider pure strategies,
but not mixed strategies.
Propositions \ref{prop:extractiveness_for_common_values_of_non-quadratic_power_technologies} and \ref{prop:extractiveness_for_observable_private_values} also hold in mixed strategies 
because
if there exists a pure-strategy Nash equilibrium,
it is also a mixed-strategy Nash equilibrium.
However,
it is not clear whether
Propositions \ref{prop:strict_extractiveness_for_common_values}, \ref{prop:strict_extractiveness_for_observable_private_values} and \ref{prop:extractiveness_for_unobservable_private_values} hold in mixed strategies
because
there might exist mixed-strategy Nash equilibria other than pure-strategy Nash equilibria.

\newpage
\appendix
\section*{Appendix}

\begin{lemma}\label{lem:bound_of_payoff}
 Let
 $f \in F$,
 $v \in V$,
 $x^\ast \in E^{f v}$
 and $i \in N$.
 Then,
 $u_i^{f v}\rnd{x^\ast} \geq 0$,
 and $f_i\rnd{x^\ast} v_i \geq x_i^\ast$.
\end{lemma}
\begin{proof}
 Because $x^\ast \in E^{f v}$,
 $u_i^{f v}\rnd{x^\ast}
  \geq u_i^{f v}\rnd{0,x_{-i}^\ast}
  = f_i\rnd{0,x_{-i}^\ast} v_i
  \geq 0
 $.
 Thus,
 $f_i\rnd{x^\ast} v_i \geq x_i^\ast$.
\end{proof}

\begin{proof}[Proof of Proposition \ref{prop:bound_of_total_effort}]
 By Lemma \ref{lem:bound_of_payoff},
 for any $i \in N$,
 $x_i^\ast \leq v_i f_i\rnd{x^\ast} \leq m^v f_i\rnd{x^\ast}$.
 Hence,
 $
  \sum_{i \in N} x_i^\ast
  \leq m^v \sum_{i \in N} f_i\rnd{x^\ast}
  = m^v
 $.
\end{proof}

\begin{lemma}\label{lem:maximization}
 Let
 $v \in \mathbb R_{> 0}$
 and $a,b \in \mathbb N$ such that
 $2 \leq b \leq a$.
 Let $x^\ast := \frac{v a \rnd{b - 1}}{b^2 \rnd{a - 1}}$.
 Let $\map{u}{\mathbb R_{\geq 0}}{\mathbb R}$ such that
 for any $x \in \mathbb R_{\geq 0}$,
 $u\rnd{x} = \frac{x^{\frac{a}{a - 1}}}{x^{\frac{a}{a - 1}} + \rnd{b - 1} \rnd{x^\ast}^{\frac{a}{a - 1}}} v - x$.
 Then,
 $x^\ast \in \arg\max_{x \in \mathbb R_{\geq 0}} u\rnd{x}$.
\end{lemma}
\begin{proof}
 Let $\map{\phi}{\mathbb R_{\geq 0}}{\mathbb R}$ such that
 for any $x \in \mathbb R_{\geq 0}$,
 \begin{align*}
  \phi\rnd{x}&
  = -
    \rnd{\rnd{2 b - 1} \rnd{x^\ast}^{\frac{a}{a - 1}} + x^{\frac{a}{a - 1}}}
    \sum_{i = 0}^{a - 1} \rnd{x^\ast}^{\frac{i}{a - 1}} x^{\frac{a - 1 - i}{a - 1}}
    +
    b^2 \rnd{x^\ast}^{\frac{2 a - 1}{a - 1}}.
 \end{align*}
 For any $x \in \mathbb R_{\geq 0}$,
 \begin{align*}
  \frac{d u\rnd{x}}{d x}
  = \frac
    {\phi\rnd{x} \rnd{x^{\frac{1}{a - 1}} - \rnd{x^\ast}^{\frac{1}{a - 1}}}}
    {\rnd{x^{\frac{a}{a - 1}} + \rnd{b - 1} \rnd{x^\ast}^{\frac{a}{a - 1}}}^2}.
 \end{align*}
 Note that
 $\phi\rnd{0} = \rnd{b - 1}^2 \rnd{x^\ast}^{\frac{2 a - 1}{a - 1}} > 0$,
 and $\phi\rnd{x^\ast} = - b \rnd{2 a - b} \rnd{x^\ast}^{\frac{2 a - 1}{a - 1}} < 0$.
 Then,
 by the intermediate value theorem,
 there exists $\bar x \in \rnd{0,x^\ast}$ such that
 $\phi\rnd{\bar x} = 0$.
 Let $x \in \mathbb R_{\geq 0}$.
 Because $\phi$ is strictly decreasing,
 $\phi\rnd{x} \gtreqless 0$
 if and only if $x \lesseqgtr \bar x$.
 Thus,
 \begin{align*}
  \frac{d u\rnd{x}}{d x}
  \begin{cases}
   \leq 0 & \text{if $x \leq \bar x$}\\
   \geq 0 & \text{if $\bar x < x \leq x^\ast$}\\
   < 0 & \text{if $x > x^\ast$}.
  \end{cases}
 \end{align*}
 Note that
 $u\rnd{0} = 0 \leq u\rnd{x^\ast}$.
 Then,
 for any $x \in \mathbb R_{\geq 0}$,
 $u\rnd{x^\ast} \geq u\rnd{x}$.
\end{proof}

\begin{lemma}\label{lem:extractiveness_for_common_values_of_power_technologies}
 Let $a \in \mathbb N$ such that $2 \leq a \leq n$.
 Let $f = f^a$.
 Then,
 $f$ is extractive for all $v \in \hat V$.
\end{lemma}
\begin{proof}
 Let $v \in \hat V$.
 Abuse $v$ as the common value
 ($v = m^v = v_i$ for any $i \in N$).
 Let $x^\ast \in X$ such that
 for some $A \in 2^N$ such that $\card{A} = a$,
 for any $i \in A$,
 $x_i^\ast = \frac{v}{a}$
 and
 for any $j \in N \setminus A$,
 $x_j^\ast = 0$.
 Let $i \in A$.
 By Lemma \ref{lem:maximization} with $b = a$,
 for any $x_i \in \mathbb R_{\geq 0}$,
 $u_i^{f v}\rnd{x^\ast} \geq u_i^{f v}\rnd{x_i,x_{-i}^\ast}$.
 Let $j \in N \setminus A$.
 For any $x_j \in \mathbb R_{> 0}$,
 \begin{align*}
  u_j^{f v}\rnd{x_j,x_{-j}^\ast}&
  = \frac{x_j^{\frac{a}{a - 1}}}{x_j^{\frac{a}{a - 1}} + a \rnd{\frac{v}{a}}^{\frac{a}{a - 1}}} v - x_j
  < \frac{x_j^{\frac{a}{a - 1}}}{x_j^{\frac{a}{a - 1}} + \rnd{a - 1} \rnd{\frac{v}{a}}^{\frac{a}{a - 1}}} v - x_j\\&
  = u_i^{f v}\rnd{x_j,x_{-i}^\ast}
  \leq u_i^{f v}\rnd{x^\ast}
  = 0
  = u_j^{f v}\rnd{x^\ast}.
 \end{align*}
 Thus,
 $x^\ast \in E^{f v}$.
 $
  \sum_{i \in N} x_i^\ast
  = a \cdot \frac{v}{a}
  = v
 $.
\end{proof}

\begin{proof}[Proof of Proposition \ref{prop:strict_extractiveness_for_common_values}]
 Let $f = f^2$.

 By Lemma \ref{lem:extractiveness_for_common_values_of_power_technologies},
 $f$ is extractive for all $v \in \hat V$.

 Let $v \in \hat V$.
 Abuse $v$ as the common value
 ($v = m^v = v_i$ for any $i \in N$).
 Let $x^\ast \in E^{f v}$.
 Let
 $A := \set{i \in N}{x_i^\ast > 0}$
 and $\alpha := \card{A}$.
 If $\alpha = 0$,
 for some $i \in N$,
 $
  u_i^{f v}\rnd{x^\ast}
  = \frac{v}{n}
  < \frac{\rnd{2 n - 1} v}{2 n}
  = u_i^{f v}\rnd{\frac{v}{2 n},x_{-i}^\ast}
 $,
 which contradicts that $x^\ast \in E^{f v}$.
 If $\alpha = 1$,
 for some $i \in A$,
 $
  u_i^{f v}\rnd{x^\ast}
  = v - x_i^\ast
  < v - \frac{x_i^\ast}{2}
  = u_i^{f v}\rnd{\frac{x_i^\ast}{2},x_{-i}^\ast}
 $,
 which contradicts that $x^\ast \in E^{f v}$.
 Thus,
 $\alpha \geq 2$.
 Let $i \in \arg\max_{j \in A} x_j^\ast$.
 For any $k \in A$,
 \begin{align*}
  0
  = \frac{\partial u_k^{f v}}{\partial x_k}\rnd{x^\ast}
  = \frac{2 v x_k^\ast \rnd{\sum_{l \in A} \rnd{x_l^\ast}^2 - \rnd{x_k^\ast}^2}}{\rnd{\sum_{l \in A} \rnd{x_l^\ast}^2}^2} - 1.
 \end{align*}
 Thus,
 for any $k \in A \setminus \crly{i}$ such that $x_k^\ast < x_i^\ast$,
 $
  x_i^\ast \rnd{\sum_{l \in A} \rnd{x_l^\ast}^2 - \rnd{x_i^\ast}^2}
  = x_k^\ast \rnd{\sum_{l \in A} \rnd{x_l^\ast}^2 - \rnd{x_k^\ast}^2}
 $,
 $
  \rnd{x_i^\ast - x_k^\ast} \rnd{\sum_{l \in A \setminus \crly{i,k}} \rnd{x_l^\ast}^2 - x_i^\ast x_k^\ast}
  = 0
 $,
 $\sum_{l \in A \setminus \crly{i,k}} \rnd{x_l^\ast}^2 = x_i^\ast x_k^\ast$.
 Suppose that
 for some $j \in A \setminus \crly{i}$,
 $x_j^\ast < x_i^\ast$
 (assumption for contradiction).
 Then,
 $\sum_{l \in A \setminus \crly{i,j}} \rnd{x_l^\ast}^2 = x_i^\ast x_j^\ast$.
 For any $k \in A \setminus \crly{i,j}$,
 $
  \rnd{x_k^\ast}^2
  \leq \sum_{l \in A \setminus \crly{i,j}} \rnd{x_l^\ast}^2
  = x_i^\ast x_j^\ast
  < \rnd{x_i^\ast}^2
 $,
 and
 thus,
 $x_k^\ast < x_i^\ast$.
 Hence,
 for any $k \in A \setminus \crly{i}$,
 $x_k^\ast < x_i^\ast$.
 Thus,
 for any $k \in A \setminus \crly{i}$,
 $x_i^\ast x_j^\ast + \rnd{x_j^\ast}^2 = \sum_{l \in A \setminus \crly{i}} \rnd{x_l^\ast}^2 = x_i^\ast x_k^\ast + \rnd{x_k^\ast}^2$,
 $\rnd{x_j^\ast - x_k^\ast} \rnd{x_j^\ast + x_k^\ast + x_i^\ast} = 0$,
 and $x_j^\ast = x_k^\ast$.
 Thus,
 $\rnd{\alpha - 2} \rnd{x_j^\ast}^2 = x_i^\ast x_j^\ast$,
 and $x_i^\ast = \rnd{\alpha - 2} x_j^\ast$.
 Hence,
 $
  0
  = \frac{\partial u_j^{f v}}{\partial x_j}\rnd{x^\ast}
  = \frac{2 v \rnd{\alpha - 1} \rnd{\alpha - 2} \rnd{x_j^\ast}^3}{\rnd{\rnd{\alpha^2 - 3 \alpha + 3} \rnd{x_j^\ast}^2}^2} - 1
 $,
 $
  x_j^\ast
  = \frac{2 v \rnd{\alpha - 1}\rnd{\alpha - 2}}{\rnd{\alpha^2 - 3 \alpha + 3}^2}
 $,
 and
 $
  \alpha
  \geq 3
 $.
 Thus,
 $
  u_j^{f v}\rnd{x^\ast}
  = \frac{\rnd{x_j^\ast}^2}{\rnd{\rnd{\alpha - 2} x_j^\ast}^2 + \rnd{\alpha - 1} \rnd{x_j^\ast}^2} v - x_j^\ast
  = - \frac{v \rnd{\alpha \rnd{\alpha - 3} + 1}}{\rnd{\alpha^2 - 3 \alpha + 3}^2}
  < 0
 $,
 which contradicts Lemma \ref{lem:bound_of_payoff}.
 Hence,
 for any $j \in A$,
 $x_j^\ast = x_i^\ast$.
 Thus,
 $
  0
  = \frac{\partial u_i^{f v}}{\partial x_i}\rnd{x^\ast}
  = \frac{2 v \rnd{\alpha - 1}}{\alpha^2 x_i^\ast} - 1
 $,
 and
 $
  x_i^\ast
  = \frac{2 v \rnd{\alpha - 1}}{\alpha^2}
 $.
 Hence,
 $
  u_i^{f v}\rnd{x^\ast}
  = \frac{v \rnd{2 - \alpha}}{\alpha^2}
 $.
 Thus,
 by Lemma \ref{lem:bound_of_payoff},
 $\alpha = 2$.
 Hence,
 $x_i^\ast = \frac{v}{2}$.
 Thus,
 for any $j \in A$,
 $x_j^\ast = \frac{v}{2}$.
 Hence,
 $
  \sum_{i \in N} x_i^\ast
  = 2 \cdot \frac{v}{2}
  = v
 $.
\end{proof}

\begin{proof}[Proof of Proposition \ref{prop:extractiveness_for_common_values_of_non-quadratic_power_technologies}]
 By Lemma \ref{lem:extractiveness_for_common_values_of_power_technologies},
 $f$ is extractive for $v$.

 Abuse $v$ as the common value
 ($v = m^v = v_i$ for any $i \in N$).
 Let $x^\ast$ be a strategy tuple such that
 for some $A \in 2^N$ such that $\card{A} = a - 1$,
 for any $i \in A$,
 $x_i^\ast = \frac{v a \rnd{a - 2}}{\rnd{a - 1}^3}$
 and
 for any $j \in N \setminus A$,
 $x_j^\ast = 0$.
 Let $i \in A$.
 By Lemma \ref{lem:maximization} with $b = a - 1$,
 for any $x_i \in \mathbb R_{\geq 0}$,
 $u_i^{f v}\rnd{x^\ast} \geq u_i^{f v}\rnd{x_i,x_{-i}^\ast}$.
 Let $j \in N \setminus A$.
 For any $x_j \in \mathbb R_{> 0}$,
 \begin{align*}
  u_j^{f v}\rnd{x_j,x_{-j}^\ast}&
  = \frac
    {
     x_j
     \rnd{x_j^{\frac{1}{a - 1}} v - x_j^{\frac{a}{a - 1}} - \rnd{a - 1} \rnd{\frac{v a \rnd{a - 2}}{\rnd{a - 1}^3}}^{\frac{a}{a - 1}}}
    }
    {
     x_j^{\frac{a}{a - 1}} + \rnd{a - 1} \rnd{\frac{v a \rnd{a - 2}}{\rnd{a - 1}^3}}^{\frac{a}{a - 1}}
    }\\
  \frac{d \rnd{v x_j^{\frac{1}{a - 1}} - x_j^{\frac{a}{a - 1}}}}{d x_j}&
  = \frac{a}{a - 1} x_j^{\frac{2 - a}{a - 1}} \rnd{\frac{v}{a} - x_j}.
 \end{align*}
 Thus,
 for any $x_j \in \mathbb R_{\geq 0}$,
 \begin{align*}
  u_j^{f v}\rnd{x_j,x_{-j}^\ast}&
  \leq \frac
       {
        x_j
        \rnd{\rnd{\frac{v}{a}}^{\frac{1}{a - 1}} v - \rnd{\frac{v}{a}}^{\frac{a}{a - 1}} - \rnd{a - 1} \rnd{\frac{v a \rnd{a - 2}}{\rnd{a - 1}^3}}^{\frac{a}{a - 1}}}
       }
       {
        x_j^{\frac{a}{a - 1}} + \rnd{a - 1} \rnd{\frac{v a \rnd{a - 2}}{\rnd{a - 1}^3}}^{\frac{a}{a - 1}}
       }\\&
  = -
    \frac
    {
     \rnd{a - 1}
     x_j
     \rnd{\frac{v}{a}}^{\frac{a}{a - 1}}
     \rnd{\rnd{1 + \frac{a \rnd{a - 3} + 1}{\rnd{a - 1}^3}}^{\frac{a}{a - 1}} - 1}
    }
    {
     x_j^{\frac{a}{a - 1}} + \rnd{a - 1} \rnd{\frac{v a \rnd{a - 2}}{\rnd{a - 1}^3}}^{\frac{a}{a - 1}}
    }
  \leq 0
  =u_j^{f v}\rnd{x^\ast}.
 \end{align*}
 Thus,
 $x^\ast \in E^{f v}$.
 $
  \sum_{i \in N} x_i^\ast
  = \rnd{a - 1} \cdot \frac{v a \rnd{a - 2}}{\rnd{a - 1}^3}
  = \frac{\rnd{a - 1}^2 - 1}{\rnd{a - 1}^2} v
  \neq v
 $.
 Thus,
 $f$ is not strictly extractive for $v$.
\end{proof}

\begin{lemma}\label{lem:extractiveness_for_observable_private_values_with_three_or_more_players}
 Let $v \in V$.
 Suppose that
 $n \geq 3$.
 Let $f \in F$ such that
 for some distinct $i,j,k \in N$ such that $i \in M^v$,
 for any $x \in X$,
 (i)
 if $x_i = m^v$,
 then $f_i\rnd{x} = 1$,
 (ii)
 if $x_i \neq m^v$ and $x_j > 0$,
 then $f_j\rnd{x} = 1$,
 and
 (iii)
 if $x_i \neq m^v$ and $x_j = 0$,
 then $f_k\rnd{x} = 1$.
 Then,
 $f$ is strictly extractive for $v$.
\end{lemma}
\begin{proof}
 Let $x^\ast \in X$ such that
 $x_i^\ast = m^v$
 and
 for any $l \in N \setminus \crly{i}$,
 $x_l^\ast = 0$.
 Then,
 $\sum_{l \in N} x_l^\ast = m^v$.
 It suffices to show that
 $E^{f v} = \crly{x_i^\ast}$.

 For any $x_i \in \mathbb R_{\geq 0} \setminus \crly{x_i^\ast}$,
 $
  u_i^{f v}\rnd{x^\ast}
  = 0
  \geq - x_i
  = u_i^{f v}\rnd{x_i,x_{-i}^\ast}
 $.
 For any $l \in N \setminus \crly{i}$ and any $x_l \in \mathbb R_{\geq 0} \setminus \crly{x_l^\ast}$,
 $
  u_l^{f v}\rnd{x^\ast}
  = 0
  \geq - x_l
  = u_l^{f v}\rnd{x_l,x_{-l}^\ast}
 $.
 Thus,
 $x^\ast \in E^{f v}$.

 Let $x \in X \setminus \crly{x^\ast}$.
 If $x_i = m^v$,
 then
 for some $l \in N \setminus \crly{i}$,
 $x_l > 0$,
 and
 $
  u_l^{f v}\rnd{x}
  = - x_l
  < 0
  = u_l^{f v}\rnd{0,x_{-l}}
 $.
 If $x_i \neq m^v$ and $x_j > 0$,
 then
 $
  u_j^{f v}\rnd{x}
  = v_j - x_j
  < v_j - \frac{x_j}{2}
  = u_j^{f v}\rnd{\frac{x_j}{2},x_{-j}}
 $.
 If $x_i \neq m^v$ and $x_j = 0$,
 then
 $
  u_j^{f v}\rnd{x}
  = 0
  < \frac{v_j}{2}
  = u_j^{f v}\rnd{\frac{v_j}{2},x_{-j}}
 $.
 Thus,
 $x \notin E^{f v}$.
\end{proof}

\begin{lemma}\label{lem:extractiveness_for_observable_private_values_with_two_heterogeneous_players}
 Let $v \in V$.
 Suppose that
 $n = 2$
 and $v \notin \hat V$.
 (i)
 Let $f \in F$ such that
 for some $i \in M^v$,
 for any $x \in X$,
 $f_i\rnd{x} = \mathbf 1_{x_i = m^v}$.
 Then,
 $f$ is extractive for $v$.
 (ii)
 Let $f \in F$.
 $f$ is not strictly extractive for $v$.
\end{lemma}
\begin{proof}
 (i)
 Let $j \in N \setminus \crly{i}$.
 Let $x^\ast \in X$ such that
 $x_i^\ast = m^v$,
 and $x_j^\ast = 0$.
 For any $x_i \in \mathbb R_{\geq 0} \setminus \crly{x_i^\ast}$,
 $
  u_i^{f v}\rnd{x^\ast}
  = 0
  \geq - x_i
  = u_i^{f v}\rnd{x_i,x_{-i}^\ast}
 $.
 For any $x_j \in \mathbb R_{\geq 0} \setminus \crly{x_j^\ast}$,
 $
  u_j^{f v}\rnd{x^\ast}
  = 0
  \geq - x_j
  = u_i^{f v}\rnd{x_j,x_{-j}^\ast}
 $.
 Thus,
 $x^\ast \in E^{f v}$.
 $\sum_{k \in N} x_k^\ast = m^v$.

 (ii)
 Let $i,j \in N$ such that $v_i > v_j$.
 Suppose that
 $f$ is extractive for $v$.
 Then,
 there exists $x^\ast \in E^{f v}$ such that
 $x_i^\ast + x_j^\ast = m^v$.
 By Lemma \ref{lem:bound_of_payoff},
 $
  m^v
  = x^\ast_i + x_j^\ast
  \leq f_i\rnd{x^\ast} m^v + f_j\rnd{x^\ast} v_j
 $.
 Thus,
 $f_j\rnd{x^\ast} \rnd{m^v - v_j} \leq 0$.
 Hence,
 $f_j\rnd{x^\ast} = 0$.
 Thus,
 by Lemma \ref{lem:bound_of_payoff},
 $x_j^\ast = 0$.
 Thus,
 $x_i^\ast = m^v$.
 Hence,
 $
  0
  = u_i^{f v}\rnd{x^\ast}
  \geq u_i^{f v}\rnd{0,x_{-i}^\ast}
  = f_i\rnd{0,x_{-i}^\ast} v_i
 $.
 Thus,
 $f_i\rnd{0,x_{-i}^\ast} = 0$.
 Hence,
 $f_i\rnd{0,0} = 0$.
 Let $y^\ast \in X$ such that
 $y_i^\ast = y_j^\ast = 0$.
 Because $x^\ast \in E^{f v}$ and $x_j^\ast = y_j^\ast$,
 for any $y_i \in \mathbb R_{\geq 0}$,
 $
  u_i^{f v}\rnd{y^\ast}
  = f_i\rnd{0,0} v_i
  = 0
  = u_i^{f v}\rnd{x^\ast}
  \geq u_i^{f v}\rnd{y_i,x_{-i}^\ast}
  = u_i^{f v}\rnd{y_i,y_{-i}^\ast}
 $;
 for any $y_j \in \mathbb R_{\geq 0}$,
 $
  u_j^{f v}\rnd{y^\ast}
  = f_j\rnd{0,0} v_j
  = v_j
  \geq f_j\rnd{y_j,y_{-j}^\ast} v_j - y_j
  = u_j^{f v}\rnd{y_j,y_{-j}^\ast}
 $.
 Thus,
 $y^\ast \in E^{f v}$.
 $y_i^\ast + y_j^\ast = 0 \neq m^v$.
 Thus,
 $f$ is not strictly extractive for $v$.
\end{proof}

\begin{proof}[Proof of Proposition \ref{prop:extractiveness_for_observable_private_values}]
 The conclusion follows from Proposition \ref{prop:strict_extractiveness_for_common_values} and Lemmas \ref{lem:extractiveness_for_observable_private_values_with_three_or_more_players} and \ref{lem:extractiveness_for_observable_private_values_with_two_heterogeneous_players}.
\end{proof}

\begin{proof}[Proof of Proposition \ref{prop:strict_extractiveness_for_observable_private_values}]
 The conclusion follows from Proposition \ref{prop:strict_extractiveness_for_common_values} and Lemmas \ref{lem:extractiveness_for_observable_private_values_with_three_or_more_players} and \ref{lem:extractiveness_for_observable_private_values_with_two_heterogeneous_players}.
\end{proof}

\begin{proof}[Proof of Proposition \ref{prop:extractiveness_for_unobservable_private_values}]
 Suppose that
 there exists $f \in F$ that is extractive for all $v \in V$
 (assumption for contradiction).

 Let $v \in V$ such that
 for some $i \in N$,
 for any $j \in N \setminus \crly{i}$,
 $v_i > v_j$.
 Let $x^\ast \in E^{f v}$ such that
 $\sum_{j \in N} x_j^\ast = v_i$.
 By Lemma \ref{lem:bound_of_payoff},
 $
  v_i
  = \sum_{j \in N} x_j^\ast
  \leq \sum_{j \in N} v_j f_j\rnd{x^\ast}
 $.
 Thus,
 $f_i\rnd{x^\ast} = 1$,
 and
 for any $j \in N \setminus \crly{i}$,
 $f_j\rnd{x^\ast} = 0$.
 Hence,
 by Lemma \ref{lem:bound_of_payoff},
 for any $j \in N \setminus \crly{i}$,
 $x_j^\ast = 0$,
 and
 $x_i^\ast = v_i$.

 Let $v,w \in \mathbb R_{> 0}^N$ such that
 for some $i \in N$,
 $v_i = 1$ and $w_i = 2$,
 and
 $v_j < v_i$ and $w_j < w_i$ for any $j \in N \setminus \crly{i}$.
 Then,
 by the assumption for contradiction,
 there exist $x^\ast \in E^{f v}$ and $y^\ast \in E^{f w}$ such that
 $\sum_{j \in N} x_j^\ast = v_i$ and $\sum_{j \in N} y_j^\ast = w_i$.
 Thus,
 $x_i^\ast = 1$ and $y_i^\ast = 2$,
 and
 for any $j \in N \setminus \crly{i}$.
 $x_j^\ast = y_j^\ast = 0$.
 Moreover,
 $f_i\rnd{x^\ast} = f_i\rnd{y^\ast} = 1$.
 Thus,
 $
  u_i^{f w}\rnd{1,y_{-i}^\ast}
  = 2 f_i\rnd{1,y_{-i}^\ast} - 1
  = 2 f_i\rnd{x^\ast} - 1
  = 1
  > 0
  = u_i^{f w}\rnd{y^\ast}
 $,
 which contradicts that $y^\ast \in E^{f w}$.
\end{proof}

\newpage
\bibliographystyle{plainnat}
\bibliography{Extractive_Contest_Design_13}

\end{document}